\documentclass{llncs}

\usepackage{amsmath,amssymb}
\usepackage{xcolor}
\usepackage{graphicx}
\usepackage{booktabs}
\usepackage{caption}
\usepackage{enumerate}
\usepackage{enumitem}

\let\doendproof\endproof
\renewcommand\endproof{~\hfill$\qed$\doendproof}

\newcommand{\REMOVE}[1]{}

\newcommand{\eps}{\varepsilon}
\newcommand{\opt}{{\rm OPT}}

\newenvironment{repeatlemma} [1]
{\noindent {\bf Lemma~\ref{#1}.}\ \slshape} {\normalfont}

\newenvironment{repeatcorollary} [1]
{\noindent {\bf Corollary~\ref{#1}.}\ \slshape} {\normalfont}

\begin{document}

\pagestyle{plain}

\title{Multi-Colored Spanning Graphs}

\author{Hugo A. Akitaya\inst{1}
	\and Maarten L\"offler\inst{2}
	\and Csaba D. T\'oth\inst{1,3}
    }

\institute{Tufts University, Medford, MA, USA\\
	\email{hugo.alves\_akitaya@tufts.edu}
	\and
	Utrecht University, Utrecht, The Netherlands\\
	\email{m.loffler@uu.nl}
    \and
    California State University Northridge, Los Angeles, CA, USA\\
    \email{csaba.toth@csun.edu}
    }

	\maketitle
	
	\begin{abstract}
	We study a problem proposed by Hurtado et al.~\cite{HKK16} motivated by sparse set visualization.
	Given $n$ points in the plane, each labeled with one or more primary colors, a \emph{colored spanning graph} (CSG)
    is a graph such that for each primary color, the vertices of that color induce a connected subgraph.
	The \textsc{Min-CSG} problem asks for the minimum sum of edge lengths in a colored spanning graph.
    We show that the problem is NP-hard for $k$ primary colors when $k\ge 3$ and provide a $(2-\frac{1}{3+2\varrho})$-approximation algorithm for $k=3$ that runs in polynomial time, where $\varrho$ is the Steiner ratio.
    Further, we give a $O(n)$ time algorithm in the special case that the input points are collinear and $k$ is constant.
	\end{abstract}
	
	\section{Introduction}
	\label{sec:intro}
	Visualizing set systems is a basic problem in data visualization.
	Among the oldest and most popular set visualization tools are the Venn and Euler diagrams.
	However, other methods are preferred when the data involves a large number of sets with complex intersection relations~\cite{AM14}.
	In particular, a variety of tools have been proposed for set systems where the elements are associated with location data.
	Many of these methods use geometric graphs to represent set membership, motivated by reducing the amount of ink used in the representation, including LineSets~\cite{AR11}, Kelp Diagrams~\cite{DK12} and KelpFusion~\cite{MH13}.

	\begin{figure}[t]
		\centering
		\includegraphics{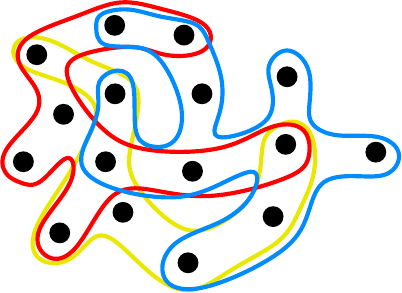}
		\hspace{2cm}
		\includegraphics{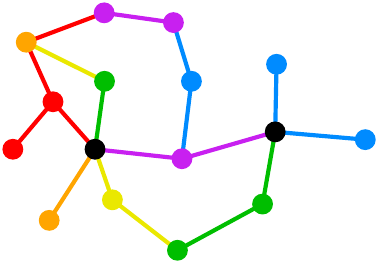}
		\caption{Left: A set of points and three subsets, $S_1$, $S_2$, and $S_3$, drawn as outlines in different colours.
        Right: The corresponding (minimum) coloured spanning graph. Refer to Section~\ref {sec:prelim} for an explanation of colour use.}
		\label{fig:example}
	\end{figure}	
	
Hurtado et al.~\cite{HKK16} recently proposed a method for drawing sets using outlines that minimise the total visual clutter. The underlying combinatorial problem is to compute a \emph{minimum colored spanning graph};
see Figure~\ref{fig:example}. They studied the problem for $n$ points in a plane and two sets (each point is a member of one or both sets). The output is a graph with the minimum sum of edge lengths such that the subgraph induced by each set is connected. They gave an algorithm that runs in $O(n^6)$-time,\footnote{An earlier claim that the problem was NP-hard~\cite{HKK13} turned out to be incorrect~\cite{HKK16}.} and a $(\frac{1}{2}\varrho+1)$-approximation in $O(n\log n)$ time, where $\varrho$ is the Steiner ratio (the ratio between the length of a minimum spanning tree and the length of a minimum Steiner tree). Efficient algorithms are known in two special cases: One runs in $O(n)$ time for collinear points that are already sorted~\cite{HKK16}; the other runs in $O(m^2+n)$ time for cocircular points, where $m$ is the number of points that are elements of both sets~\cite{BBD16}. This problem also has applications for connecting different networks with minimum cost, provided that edges whose endpoints belong to both networks can be shared.
	
\smallskip\noindent{\bf Results and organization.}
We study the minimum colored spanning graph problem for $n$ points in a plane and $k$ sets, $k\ge 3$.
The formal definition and some properties of the optimal solution are in Section~\ref{sec:prelim}.
In Section~\ref{sec:hard}, we show that \textsc{Min-$k$CSG} is NP-complete for all $k\geq 3$,
and in Section~\ref{sec:approximation} we provide an $(2-\frac{2}{2+2\varrho})$-approximation algorithm
for $k=3$ that runs in $O(n\log n+m^6)$ time, where $m$ is the number of multichromatic points. This improves the previous $(2+\frac{\varrho}{2})$-approximation from~\cite{HKK16}.
Section~\ref{sec:collinear} describes an algorithm for the special case of collinear points that runs in $2^{O(k^2 2^k)}\cdot n$ time.
Due to space constraints, some proofs are omitted; they can be found in Appendix~\ref {app:proofs}.

\section{Preliminaries}
\label{sec:prelim}

In this section, we define the problem and show a property of the optimal solution related to the minimum spanning trees, which is used in Sections~\ref{sec:hard} and \ref{sec:approximation}.

\smallskip\noindent{\bf Definitions.} Given a set of $n$ points in the plane $S=\{p_1,\ldots,p_n\}$ and subsets $S_1,\ldots , S_k\subseteq S$, we represent set membership with a function $\alpha:S\rightarrow 2^{\{1,\ldots,k\}}$, where $p\in S_c$ iff $c\in \alpha(p)$ for every \emph{primary color} $c\in \{1,\ldots , k\}$. We call $\alpha(p)$ the \emph{color} of point $p$. A point $p$ is \emph{monochromatic} if it is a member of a single set $S_i$, that is, $|\alpha(p)|=1$, and \emph{multi-chromatic} if $|\alpha(p)|>1$. For an edge $\{p_i,p_j\}\in E$ in a graph $G=(S,E)$, we use the shorthand notation $\alpha(\{p_i,p_j\})=\alpha(p_i)\cap \alpha(p_j)$ for the shared primary colors of the two vertices. For every $c\in \{1,\ldots ,k\}$, we let $G_c=(S_c,E_c)$ denote the subgraph of $G=(S,E)$ induced by $S_c$.
%
%
All figures in this paper depict only three primary colors: \texttt{r}, \texttt{b}, and \texttt{y} for red, blue, and yellow respectively.
Multi-chromatic points and edges are shown green, orange, purple, or black if their color is $\{\texttt{b}, \texttt{y}\}$,
$\{\texttt{r}, \texttt{y}\}$ or $\{\texttt{r}, \texttt{b}\}$, or $\{\texttt{r}, \texttt{b}, \texttt{y}\}$, respectively. See, for example, Fig.~\ref {fig:example} (b).

A \emph{colored spanning graph} for the pair $(S,\alpha)$, denoted CSG$(S,\alpha)$, is a graph $G=(S,E)$ such that $(S_c,E_c)$ is connected for every primary color $c\in \{1,\ldots ,k\}$.	The \emph{minimum colored spanning graph} problem (\textsc{Min-CSG}), for a given pair $(S,\alpha)$, asks for the minimum cost $\sum_{e\in E} w(e)$ of a CSG$(S,\alpha)$, where $w(e)$ is the Euclidean length of $e$.
When we wish to emphasize the number $k$ of primary colors, we talk about the \textsc{Min-$k$CSG} problem.

\smallskip\noindent{\bf Monochromatic edges in a minimum CSG.}
The following lemma shows that we can efficiently compute some of the monochromatic edges of a minimum CSG for an instance $(S,\alpha)$ using the \emph{minimum spanning tree} (\emph{MST}) of $S_c$ for every primary color $c\in\{1,\ldots , k\}$.	

\begin{lemma}\label{lem:monocrhomatic}
	Let $(S,\alpha)$ be an instance of \textsc{Min-CSG} and $c\in \{1,\ldots , k\}$.
	Let $E(MST(S_c))$ be the edge set of an MST of $S_c$, and let $S_c'$ be the set of multi-chromatic points in $S_c$.
	Then there exists a minimum CSG that contains at least $|E($MST$(S_c))|-|S_c'|+1$ edges of $E(MST(S_c))$.
	The common edges of $E(MST(S_c))$ and of such a minimum CSG can be computed in $O(n\log n)$ time.
\end{lemma}
\begin{proof}
Construct a monochromatic subset $E_c'\subset E($MST$(S_c))$ by successively removing a longest edge from the path in MST$(S_c)$ between any two points in $S_c'$. An MST$(S_c)$ can be computed in $O(n\log n)$ time, and $E_c'$ can be obtained in $O(n)$ time. The graph $(S_c,E_c')$ has $|S_c'|$ components, each containing one element of $S_c'$, hence $|E_c'|=|E($MST$(S_c))|-|S_c'|+1$.
	
Let $(S,E^{\opt})$ be a minimum CSG. While there is an edge $e\in E_c'\setminus E^{\opt}$, we can find an edge $e^*\in E^{\opt}\setminus E_c'$ such that exchanging $e^*$ for $e$ yields another minimum CSG. Indeed, since $(S_c,E^\opt_c)$ is connected, the insertion of the edge $e$ creates a cycle $C$ that contains $e$. Consider the longest (open or closed) path $P\subseteq C$ that is monochromatic and contains $e$. Note that at least one of the endpoints of $e$ is monochromatic, therefore $P$ contains at least two monochromatic edges. Since every component of $(S_c,E_c')$ is a tree and contains only one multi-chromatic point, there is a monochromatic edge $e^*\in E^\opt\setminus E_c'$ in $P$.
We have $w(e)\le w(e^*)$, because there is a cut of the complete graph on $S_c$	that contains both $e$ and $e^*$, and $e\in E($MST$(S_c))$. Since $\alpha(e^*)=c$, the deletion of $e^*$ can only influence the connectivity of the induced subgraph $(S_c,E^{\opt}_c)$. Consequently, $(S,E^{\opt}\cup \{e\}\setminus\{e^*\})$ is a CSG with equal or lower cost than $(S,E^{\opt})$.
By successively exchanging the edges in $E_c'\setminus E^{\opt}$, we obtain a minimal CSG containing $E_c'$.
\end{proof}

Hurtado et al.~\cite{HKK16} gave an $O(n^6)$-time algorithm for \textsc{Min-2CSG}, by a reduction to a matroid intersection problem on the set of all possible edges on $S$, which has $O(n^2)$ elements. Their algorithm for matroid intersection finds $O(n^2)$ single source shortest paths in a bipartite graph with $O(n^2)$ vertices and $O(n^4)$ edges, which leads to an overall running time of $O(n^6)$.
We improve the runtime to $O(n\log n+m^6)$, where $m$ is the number of multi-chromatic points.

\begin{corollary}\label{cor:m6}
An instance $(S,\alpha)$ of \textsc{Min-2CSG} can be solved in $O(n\log n+m^6)$ time, where $m$ is the number of multi-chromatic points in $S$.
\end{corollary}
\begin{proof}
	By Lemma~\ref{lem:monocrhomatic}, we can compute two spanning forests on $S_1$ and $S_2$, respectively, each with $m$ components, that are subgraphs of a minimum CSG in $O(n\log n)$ time. It remains to find edges of minimum total length that connect these components in each color, for which we can use the same matroid intersection algorithm as in~\cite{HKK16}, but with a ground set of size $O(m^2)$.
\end{proof}

	\section{General Case}
	\label{sec:hard}
	
	We show that the decision version of \textsc{Min-CSG} is NP-complete.
	We define the decision version of \textsc{Min-CSG} as follows: given an instance $(S,\alpha)$ and $W>0$,
    is there a CSG  $(S,E)$ such that $\sum_{e\in E}w(e)<W$?
	
	\begin{lemma}
		\label{lem:np}
		\textsc{Min-$k$CSG} is in NP.
	\end{lemma}
	\begin{proof}
Given a set of edges $E$, we can verify if $(S,E)$ is a $CSG(S,\alpha)$ in $O(k|S|)$ time by testing connectivity in $(S_c,E_c)$ for each primary color $c\in\{1,\ldots,k\}$, and then check whether $\sum_{e\in E}w(e)\le W$ in $O(|E|)$ time.
	\end{proof}

	We reduce \textsc{Min-3CSG} from \textsc{Planar-Monotone-3SAT}, which is known to be NP-complete~\cite{BK12}.
	For every instance $A$ of \textsc{Planar-Monotone-3SAT}, we construct an instance $f(A)$ of \textsc{Min-3CSG}.
	An instance $A$ consists of a plane bipartite graph between $n$ variable and $m$ clause vertices such that every clause has degree three or two, all variables lie on the $x$-axis and edges do not cross the $x$-axis.
	Clauses are called \emph{positive} if they are in the upper half-plane or \emph{negative} otherwise.
	The problem asks for an assignment from the variable set to $\{\texttt{true}, \texttt{false}\}$ such that each positive (negative) clause is adjacent to a \texttt{true} (\texttt{false}) variable.
	
    \begin{figure}[htp]
		\centering
		\includegraphics [width=\textwidth]{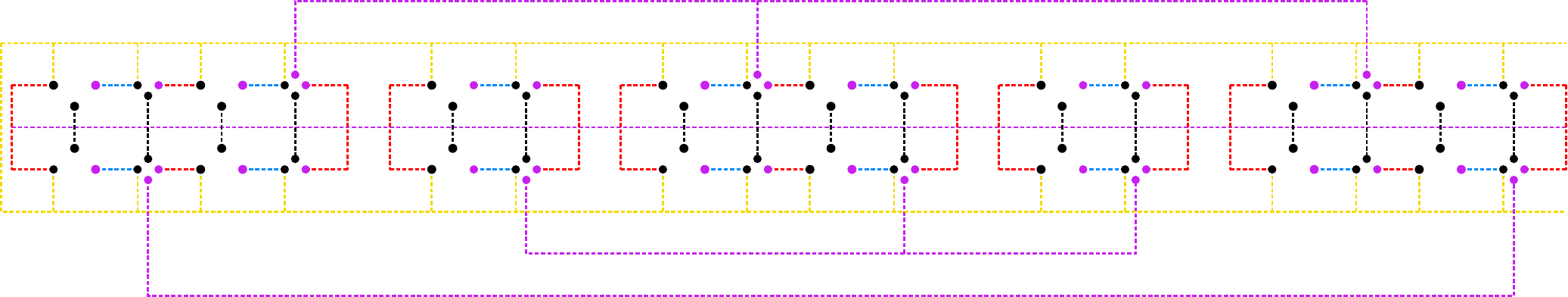}
		\caption{Construction for an instance
		$A$ equivalent to the boolean formula $(x_1 \vee x_3 \vee x_5)
		 \wedge (\neg x_1 \vee \neg x_5)
		 \wedge (\neg x_2 \vee \neg x_3 \vee \neg x_4)$.
		}
		\label{fig:reduction-complete}
	\end{figure}

	Given an instance $A$ of \textsc{Planar-Monotone-3SAT}, we construct $f(A)$ as shown in Fig.~\ref {fig:reduction-complete}
 (a single variable gadget is shown in Fig.~\ref{fig:reduction} in the Appendix).
	The points marked with small disks are called \emph{active} and they are the only multi-chromatic points in the construction.
	The dashed lines in a primary color represent a chain of equidistant monochromatic points, where the gap between consecutive points is $\eps$.  A purple (resp., black) dashed line represents a red and a blue (resp., a red, a blue, and a yellow) dashed line that run $\eps$ close to each other.	Informally, the value of $\eps$ is set small enough such that every point in the interior of a dashed line is adjacent to its neighbors in any minimum CSG. The boolean assignment of $A$ is encoded in the edges connecting active points.
	We break the construction down to gadgets and explain their behavior individually.

	The long horizontal purple dashed line is called \emph{spine} and the set of yellow dashed lines (shown in Fig.~\ref{fig:gadgets}(a)) is called \emph{cage}.
	The rest of the construction consists of variable and clause gadgets (shown in Figs.~\ref{fig:gadgets}(b) and (c)).	
	The width of a variable gadget depends on the degree of the corresponding variable in the bipartite graph given by the instance $A$.
	For every edge incident to the variable, we repeat the middle part of the gadget as shown in Fig.~\ref{fig:gadgets}(b)
   (cf. Fig.~\ref{fig:reduction}, where a variable of degree-2 is shown).
	The vertical black dashed lines are called \emph{ribs} and the set of three or four active points close to an endpoint of a rib is called \emph{switch}.
	The variable gadget contains switches of two different sizes alternately from left to right.
	A \emph{2-switch} (resp., \emph{2$\delta$-switch}) is a switch in which active points are at most 2 (resp., $2\delta$) apart.
	The clause gadgets are positioned as the embedding of clauses in $A$; refer to Fig.~\ref {fig:reduction-complete}.
	Each active point of a positive (negative) clause is assigned to a $2\delta$-switch and positioned vertically above (below) the active point of the rib, at distance  $2\delta$ from it.
	
	Let $E'$ be the set of all monochromatic edges of a minimum CSG computable by Lemma~\ref{lem:monocrhomatic}.
	Let $r$ be the number of edges in the bipartite graph of $A$.
	The instance $f(A)$ contains $13r$ active points, so $(S,E')$ contains $13r$ connected components.
	By construction, the number of $\eps$-edges in a solution of $f(A)$ between components of $(S,E')$ is upper bounded by $39r$ (one edge per color per component).
	Finally, we set $W=(\sum_{e\in E'}w(e))+39r\eps+r(2+2\sqrt{2})+r\delta(2+2\sqrt{2})+m\delta(2\sqrt{2}-2)$ and we choose $\eps= \frac{1}{500r^2}$ and  $\delta=\frac{1}{10r}$.
	This particular choice of $\eps$ and $\delta$ is justified by the proofs of Corollaries~\ref{cor:2-pair} and \ref{cor:clause}.
	By construction, $f(A)$ has the following property:

    \begin{enumerate}
	\item[(I)] For every partition of the components of $(S_c,E'_c)$ into two sets $C_1, C_2$, where $c$ is a primary color, let $\{p_1,p_2\}$ be the shortest edge between $C_1$ and $C_2$. Then either $w(\{p_1,p_2\})=\eps$ or $p_1$ and $p_2$ are active points in the same switch.
	\end{enumerate}
	
	\begin{figure}[t]
		\centering
		\def\svgwidth{0.75\columnwidth}
		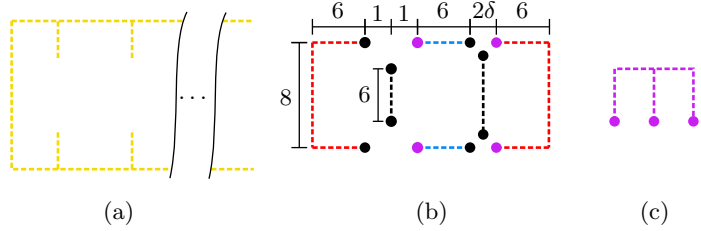
		\caption{(a) Cage. (b) Variable gadget. (c) Clause gadget.}
		\label{fig:gadgets}
	\end{figure}
	
	\begin{definition}
A \emph{standard solution} of \textsc{Min-3CSG} is a solution that contains $E'$ and in  which every edge longer than $\eps$ is between two  active points of the same switch.
	\end{definition}
		
	\begin{lemma}
		\label{lem:sat2csg}
		Let $A$ be a positive instance of \textsc{Planar-Monotone-3SAT}. Then $f(A)$ is a positive instance of \textsc{Min-3CSG}.
	\end{lemma}
	
    To prove the lemma, we construct a standard solution for $f(A)$ based on the solution for $A$. This proof, and subsequent proofs, argues about all possible ways to connect the vertices in a switch of $f(A)$. The most efficient ones are shown in Fig.~\ref {fig:cases}; these may appear in an optimal solution. Refer to Fig.~\ref {fig:switch} in the Appendix for a full list.	
	
	
	\begin{figure}[h]
		\centering	
			\def\svgwidth{0.9\columnwidth}
			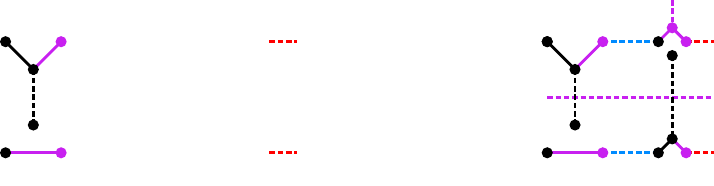
		\caption{Possible ways to connect the vertices in a switch of $f(A)$. (a) One of the two states of a 2-switch, encoding the truth value of the variable. (b) The two possible states of a $2\delta$-switch if the incident clause is not satisfied through this variable. (c) The only possible state of a $2\delta$-switch if the incident clause is satisfied through this variable.}
		\label{fig:cases}
	\end{figure}

	\begin{lemma}
		\label{lem:existance}
		If $f(A)$ is a positive instance of \textsc{Min-3CSG}, there exists a standard solution for this instance.
	\end{lemma}
	
	Before proving the other direction of the reduction, we show some properties of a standard solution.
	The active points in a switch impose some local constraints.
	The black and purple points attached to horizontal dashed lines determine the \emph{switch constraint}:
	since these points have more colors than their incident dashed lines, they each are incident to at least one edge in the switch.
	Each rib determines a \emph{rib constraint} to a pair of switches that contain its endpoints: at least one of these switches must contain an edge between its black active points or else there is no yellow path between this rib and the cage.
	The following lemmas establish some bounds on the length of the edges used to satisfy local constraints of a pair of switches adjacent to a rib.
	We refer to this pair as a $2$-pair or $2\delta$-pair according to the type of the switch.

	\begin{lemma}
		\label{lem:min}
		In a standard solution, the minimum length required to satisfy the local constraints of a 2-pair (resp., $2\delta$-pair) is $2(1+\sqrt{2})$ (resp., $2\delta(1+\sqrt{2})$).
	\end{lemma}
	
	\begin{corollary}
		\label{cor:2-pair}
		In a standard solution, every 2-pair is connected minimally.
	\end{corollary}
		
	\begin{lemma}
		\label{lem:clause}
		In a standard solution, for each clause gadget, there exists a $2\delta$-pair with local cost at least $4\delta\sqrt{2}$.
	\end{lemma}

	\begin{corollary}
		\label{cor:clause}
		In a standard solution, for each clause gadget, there exists a $2\delta$-pair connected as Fig.~\ref{fig:cases}(c). All other $2\delta$-pairs are connected minimally as shown in Fig.~\ref{fig:cases}(b).
	\end{corollary}

	\begin{lemma}
		\label{lem:csg2sat}
		Let $f(A)$ be a positive instance of \textsc{Min-3CSG}. Then $A$ is a positive instance of \textsc{Planar-Monotone-3SAT}.
	\end{lemma}

	The following theorem is a direct consequence of Lemmata~\ref{lem:np}, \ref{lem:sat2csg}, and \ref{lem:csg2sat}.
	
	\begin{theorem}
		\textsc{Min-$k$CSG} is NP-complete for $k\ge3$.
	\end{theorem}

\section{Approximation}
\label{sec:approximation}
%

Hurtado et al.~\cite{HKK16} gave an approximation algorithm for \textsc{Min-$k$CSG} that runs in $O(n\log n)$ time and achieves a ratio of $\lceil k/2\rceil +\lfloor  k/2\rfloor\varrho/2$, where $\varrho$ is the Steiner ratio.
The value of $\varrho$ is not known and the current best upper bound is $\varrho\leq 1.21$ by Chung and Graham~\cite{CG86} (Gilbert and Pollack~\cite{GP68} conjectured $\varrho =\frac{2}{\sqrt{3}}\approx 1.15$).
For the special case $k=3$, the previous best approximation ratio is $2+\varrho/2\leq 2.6$.
We improve the approximation ratio to 2, and then further to 1.81.
Our first algorithm immediately generalises to $k \ge 3$, and yields an $\lceil k/2 \rceil$-approximation, improving on the general result by Hurtado et al.; our second algorithm also generalizes to $k>3$, however, we do not know whether it achieves a good ratio.

Suppose we are given an instance of \textsc{Min-3CSG} defined by $(S,\alpha)$  where $|S|=n$ and the set of primary colors is  $\{$\texttt{r,b,y}$\}$.
We define $\alpha_{\texttt{rb}}:S_\texttt{r}\cup S_\texttt{b}\rightarrow 2^{\{r,b\}}\setminus\{\emptyset\}$ where $\alpha_{\texttt{rb}}(p)=\alpha(p)\setminus\{\texttt{y}\}$.
Let $G^*$ be an optimal solution for \textsc{Min-3CSG}, and put $\opt=\|G^*\|$.
Algorithm A1 computes a minimum red-blue-purple graph $G_\texttt{rb}=CSG(S_\texttt{r}\cup S_\texttt{b},\alpha_\texttt{rb})$ in $O(n\log n+m^6)$ time, where $m=|S_\texttt{r}\cap S_\texttt{b}|$ by Corollary~\ref{cor:m6}; then computes a minimum spanning tree $G_\texttt{y}$ of $S_\texttt{y}$, and returns the union $G_{\texttt{rb}}\cup G_\texttt{y}$. Since $G^*$ contains a red, a blue, and a yellow spanning tree, we have $\|G_{\texttt{rb}}\|\leq \opt$ and $\|G_\texttt{y}\|\leq \opt$; that is, Algorithm A1 returns a solution to \textsc{Min-3CSG} whose length is at most $2\opt$.

\begin{theorem}
Algorithm A1 returns a 2-approximation for \textsc{Min-3CSG}; it runs in $O(n\log n+m^6)$ time on $n$ points, $m$ of which are multi-chromatic.
\end{theorem}

Algorithm A1 can be extended to $k$ colors by partitioning the primary colors into $\lceil\frac{k}{2}\rceil$ groups of at most two and computing the minimum CSG for each group. The union of these graphs is a $\lceil\frac{k}{2}\rceil$-approximation that can be computed in $O(kn^6)$ time.

Algorithm A2 computes six solutions for a given instance of \textsc{Min-3CSG}, $G_1,\ldots , G_6$, and returns one with minimum weight.
Graph $G_1$ is the union of $G_\texttt{rb}$ and $G_\texttt{y}$ defined above.
Graphs $G_2$ and $G_3$ are defined analogously: $G_2=G_\texttt{ry}\cup G_\texttt{b}$ and $G_3=G_\texttt{by}\cup G_\texttt{r}$, each of which can be computed in $O(n^6)$ time by~\cite{HKK16}.
Let $S_\texttt{rby}\subseteq S$ be the set of ``black'' points that have all three colors, and let $H$ be an MST of $S_\texttt{rby}$, which can be computed in $O(n\log n)$ time. We augment $H$ into a solution of  \textsc{Min-3CSG} in three different ways as follows.
First, let $G_{\texttt{rb}:H}$ be the minimum forest such that $H\cup G_{\texttt{rb}:H}$ is a minimum red-blue-purple spanning graph on $S_\texttt{r}\cup S_\texttt{b}$.
$G_{\texttt{rb}:H}$ can be computed in $O(n\log n+m^6)$ time by the same matroid intersection algorithm as in Corollary~\ref{cor:m6}, by setting the weight of any edge between components containing black points to zero. Similarly, let $G_{\texttt{y}:H}$ be the minimum forest such that $H\cup G_{\texttt{y}:H}$ is a spanning tree on $S_\texttt{y}$, which can be computed in $O(n\log n)$ time by Prim's algorithm.
Now we let $G_4=H\cup G_{\texttt{rb}:H}\cup G_{\texttt{y}:H}$. Similarly, let $G_5=H\cup G_{\texttt{ry}:H}\cup G_{\texttt{b}:H}$ and  $G_6=H\cup G_{\texttt{by}:H}\cup G_{\texttt{r}:H}$.

\begin{theorem}
Algorithm A2 returns a $(2-\frac{1}{3+2\varrho})$-approximation for \textsc{Min-3CSG}; it runs in $O(n\log n+m^6)$ time on an input of $n$ points, $m$ of which are multi-chromatic.
\end{theorem}
\begin{proof}
Consider an instance $(S,\alpha)$ of \textsc{Min-3CSG}, and let $G^*=(S,E^*)$ be an optimal solution with $\|E^*\|=\opt$. Partition $E^*$ into $7$ subsets: for every color $\gamma\in 2^{\{\texttt{r},\texttt{b},\texttt{y}\}}\setminus \emptyset$, let $E_\gamma^*=\{e\in E^*: \alpha(e)=\gamma\}$, that is $E_\gamma^*$ is the set of edges of color $\gamma$ in $G^*$. Put $\beta=\|E_\texttt{rby}^*\|/\opt$.
Then we have $2(1-\beta)\opt
=(2\|E_\texttt{r}^*\|+\|E_\texttt{rb}^*\|+\|E_\texttt{ry}^*\|)
+(2\|E_\texttt{b}^*\|+\|E_\texttt{rb}^*\|+\|E_\texttt{by}^*\|)
+(2\|E_\texttt{y}^*\|+\|E_\texttt{ry}^*\|+\|E_\texttt{by}^*\|)$.
Without loss of generality, we may assume
$2\|E_\texttt{y}^*\|+\|E_\texttt{ry}^*\|+ \|E_\texttt{by}^*\|\leq \frac23(1-\beta)\opt$.
%

First, consider $G_1=G_\texttt{rb}\cup G_\texttt{y}$. The edges of $G^*$ whose colors include red or blue (resp., yellow)
form a connected graph on $S_\texttt{r}\cup S_\texttt{b}$ (resp., $S_\texttt{y})$. Consequently,
\begin{eqnarray}
\|G_\texttt{rb}\|
&\leq& \|E_\texttt{r}^*\|+ \|E_\texttt{b}^*\| + \| E_\texttt{rb}^*\| + \| E_\texttt{ry}^*\| + \|E_\texttt{by}^*\| +\| E_\texttt{rby}^*\|. \label{eq:Grb}\\
\|G_\texttt{y}\|
&\leq & \|E_\texttt{y}^*\|+ \| E_\texttt{ry}^*\| + \|E_\texttt{by}^*\| +\| E_\texttt{rby}^*\|.\label{eq:Gy}
\end{eqnarray}
	The combination of \eqref{eq:Grb} and \eqref{eq:Gy} yields
	\begin{eqnarray}
	\|G_1\| &\leq & \|G_\texttt{rb}\| + \|G_\texttt{y}\|
	\leq  \opt +  \| E_\texttt{ry}^*\| + \|E_\texttt{by}^*\| +\| E_\texttt{rby}^*\| \nonumber\\
	&\leq&\opt + \frac{2}{3}(1-\beta)\cdot\opt+\beta\cdot\opt
	=\frac{5+\beta}{3}\opt. \label{eq:g1}
	\end{eqnarray}
	
Next, consider $G_4=H\cup G_{\texttt{rb}:H}\cup G_{\texttt{y}:H}$. The edges of $G^*$ whose colors include yellow contain a spanning tree on $S_\texttt{y}$, hence a Steiner tree on the black points $S_\texttt{rby}$. Specifically, the black edges in $E_{\texttt{rby}}^*$ form a black spanning forest, which is completed to a Steiner tree by some of the edges of $E_\texttt{y}^*\cup E_\texttt{by}^*\cup E_\texttt{ry}^*$. This implies
	\begin{eqnarray*}
		\|H\|
		&\leq &\|E_\texttt{rby}^*\|+ \varrho\cdot (\|E_\texttt{y}^*\| +\|E_\texttt{by}^*\|+\|E_\texttt{ry}^*\|)\\
		&\leq & \beta\cdot \opt+ \varrho\frac{2}{3}(1-\beta)\cdot \opt
		= \left(\beta+\frac23 \varrho -\frac23 \beta\varrho\right)\opt.
	\end{eqnarray*}
	Since $H$ is a spanning tree on the black vertices $S_{rby}$, \eqref{eq:Grb} and \eqref{eq:Gy} reduce to
	\begin{eqnarray}
	\|G_{\texttt{rb}:H}\| &\leq& \|E_\texttt{r}^*\|+ \|E_\texttt{b}^*\| + \| E_\texttt{rb}^*\| + \| E_\texttt{ry}^*\| + \|E_\texttt{by}^*\|,\label{eq:Grb+}\\
	\|G_{\texttt{y}:H}\| &\leq& \|E_\texttt{y}^*\| + \| E_\texttt{ry}^*\| + \|E_\texttt{by}^*\|.\label{eq:Gy+}
	\end{eqnarray}
	The combination of \eqref{eq:Grb+} and \eqref{eq:Gy+} yields
	\begin{eqnarray*}
		\|G_{\texttt{rb}:H}\| +\|G_{\texttt{y}:H}\|
		&\leq& (\opt-\|E_\texttt{rby}^*\|)+ \| E_\texttt{ry}^*\| + \| E_\texttt{by}^*\|\\
		&\leq& (1-\beta)\cdot \opt + \frac23(1-\beta)\cdot \opt
        = \frac53(1-\beta)\cdot \opt.
	\end{eqnarray*}
	Therefore,
	\begin{equation}\label{eq:g4}
	\|G_4\|
	= \|H\| + \|G_{\texttt{rb}:H} \| + \|G_{\texttt{y}:H}\|
	\leq  \left(\frac53+\frac23 ( \varrho -\beta -\beta\varrho)\right)\opt.
	\end{equation}
	
	If we set $\beta=\frac{2\varrho}{3+2\varrho}$, then
	both \eqref{eq:g1} and \eqref{eq:g4} give the same upper bound
	$$\frac{\min(\|G_1\|,\|G_4\|)}{\opt} \leq \frac{5+\beta}{3}=2-\frac{1}{3+2\varrho}\leq 1.816,$$
	where we used the current best upper bound  for the Steiner ratio $\varrho\leq 1.21$ from~\cite{CG86}.
	\end{proof}

\section{Collinear points}
\label{sec:collinear}

In this section we consider instances of \textsc{Min-$k$CSG}, $(S,\alpha)$, where $k\geq 3$ and $S$ consists of collinear points. An example is shown in Fig.~\ref {fig:example-1d}.
Without loss of generality, $S=\{p_1,\ldots , p_n\}$ and the points $p_i$, $1\leq i\leq n$, lie on the $x$-axis sorted by $x$-coordinates.
We present a dynamic programming algorithm that solves \textsc{Min-$k$CSG} in $2^{O(k^2 2^k)}\cdot n$ time.

	\begin{figure}[htp]
		\centering
		\includegraphics{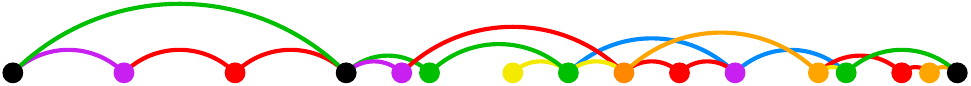}
		\caption{An example with optimal solution for collinear points.}
		\label{fig:example-1d}
	\end{figure}

Our first observation is that if the points in $S$ are collinear, we may assume that every edge satisfies the following property.
\begin{quote}
	If $\{p_a,p_b\}$, $a<b$, is an edge, then there is no $r$, $a<r<b$, such that $\alpha(\{p_a,p_b\})\subseteq \alpha(p_r)$. \hfill ($\star$)
\end{quote}

\begin{lemma}\label{lem:unique-color}
For every graph $G=(S,E)$, there exists a graph $G'=(S,E')$ of the same cost that satisfies ($\star$) and for each color $c\in \{1,\ldots, k\}$, every component of $(S_c,E_c)$ is contained in some component of $(S_c,E_c')$. In particular, \textsc{Min-$k$CSG} has a solution with property ($\star$).
\end{lemma}
\begin{proof}
	Let $G=(S,E)$ be a graph, and let $X_G$ denote the set of triples $(i,j;r)$ such that $1\leq i<r<j\leq n$, $\{p_i,p_j\}\in E$, and $\alpha(\{p_i,p_j\})\subseteq \alpha(p_r)$. If $X_G=\emptyset$, then $G$ satisfies ($\star$). Suppose $X_G\neq \emptyset$. For every triple $(i,j;r)\in X_G$, successively, replace the edge $\{p_i, p_j\}$ by two edges $\{p_i, p_r\}$ and $\{p_r, p_j\}$ (i.e., subdivide edge $\{p_i,p_j\}$ at $p_r$). Note that $\alpha(\{p_h,p_i\}),\alpha(\{p_i,p_j\})\subseteq \alpha(p_i)$, consequently $p_i$ and $p_j$ remain in the same component for each primary color $c\in \alpha(\{p_i,p_j\})$. Each step maintains the total edge length of the graph and strictly decreases $X_G$. After $|X_G|$ subdivision steps, we obtain a graph $G'=(S,E')$ as required.
\end{proof}

In the remainder of this section we assume that every edge has property ($\star$). Furthermore, all graphs considered in this section are defined on an interval of consecutive vertices of $S$.

\begin{corollary}\label{cor:unique-edge}
	Let $G=(S,E)$ be a graph and let $i\in \{1,\ldots , n\}$.
	\begin{enumerate}
		\item\label{cor:unique}  If $e\in E$ is an edge between $\{p_1,\ldots , p_i\}$ and $\{p_{i+1},\ldots, p_n\}$ and $\alpha(e)=\gamma$,
		then the endpoints of $e$ are uniquely determined.  Specifically, if $e=\{p_a,p_b\}$ with $1\leq a\leq i<b\leq n$, then
		$a\in \{1,\ldots , i\}$ is the largest index such that $\gamma\subset \alpha(p_a)$, and
		$b\in \{i+1,\ldots , n\}$ is the smallest index such that $\gamma\subset \alpha(p_b)$.
		\item\label{cor:overlap} If two edges $e_1, e_2\in E$ overlap, then $\alpha(e_1)\neq \alpha(e_2)$.
	\end{enumerate}
\end{corollary}
\begin{proof}
	\eqref{cor:unique} Suppose, to the contrary, that there is index $j$, $a<j\leq i$, such that $\gamma\subset \alpha(p_j)$.
	Then edge $\{p_a,p_b\}$ and point $p_r$ violate ($\star$). The case that there is some $j$, $i+1\leq j<b$, leads to the same contradiction.
	
	\eqref{cor:overlap} Without loss of generality $e_1=\{p_a,p_b\}$ and $e_2=\{p_i,p_j\}$ with $a\leq i< b\leq j$.
	Then both edges $e_1$ and $e_2$ are between $\{p_1,\ldots , p_i\}$ and $\{p_{i+1},\ldots , p_n\}$, contradicting part~\eqref{cor:unique}.
\end{proof}


The basis for our dynamic programming algorithm is that \textsc{Min-$k$CSG} has the \emph{optimal substructure} and \emph{overlapping substructures} properties when the points in $S$ are collinear. We introduce some notation for defining the subproblems.
For indices $1\leq a\leq b\leq n$, let $S[a,b]=\{p_a,\ldots, p_b\}$. For every graph $G=(S,E)$ and index $i\in \{1,\ldots ,n\}$, we partition the edge set $E$ into three subsets as follows: let $E_i^-$ be the set of edges induced by $S[1,i]$, $E_i^+$ the set of edges induced by $S[i+1,n]$, and $E_i^0$ the set of edges between $S[1,i]$ and $S[i+1,n]$. With this notation, \textsc{Min-$k$CSG} has the following optimal substructure property.
\begin{lemma}\label{lem:opt-subproblem}
	Let $G=(S,E)$ be a minimum CSG, $i\in \{1,\ldots , n\}$, and $\mathcal{X}$ be the family of edge sets $X_i^-$ on $S[1,i]$ such that $(S,X_i^- \cup E_i^0\cup E_i^+)$ is a CSG. Then $(S,X_i^- \cup E_i^0\cup E_i^+)$ is a minimum CSG iff $X_i^-\in \mathcal{X}$ has minimum cost.
\end{lemma}
\begin{proof}
	If $(S,X_i^- \cup E_i^0\cup E_i^+)$ is a minimum CSG, but some $Y_i^-\in \mathcal{X}$ costs less than $E_i^-$, then $(S,Y_i^-\cup E_i^0\cup E_i^+)$ would be a CSG that costs less than $G=(S,E)$, contradicting the minimality of $(S,X_i^- \cup E_i^0\cup E_i^+)$.
	If $X_i^-\in \mathcal{X}$ has minimum cost, but $G=(S,E)$ costs less than $(S,X_i^- \cup E_i^0\cup E_i^+)$, then $E_i^-\in \mathcal{X}$ would costs less than $X_i^-$, contradicting the minimality of $X_i^-\in \mathcal{X}$.
\end{proof}

Lemma~\ref{lem:opt-subproblem} immediately suggests a na\"ive algorithm for \textsc{Min-$k$CSG}: Guess the edge set $E_i^0\cup E_i^+$ of a minimum CSG $G=(S,E)$, and compute a minimum-cost set $X_i^-$ on $S[1,i]$ such that $(S,X_i^-\cup E_i^0\cup E_i^+)$ is a CSG.
However, all possible edge sets $E_i^0\cup E_i^+$ could generate $2^{\Theta(n)}$ subproblems. We reduce the number of subproblems using the overlapping subproblem property. Instead of guessing $E_i^0\cup E_i^+$, it is enough to guess the information relevant for finding the minimal cost $X_i^-$ on $S[1,i]$. First, the edges in $E_i^0$ can be uniquely determined by the set of their colors (using Corollary~\ref{cor:unique-edge}\eqref{cor:unique}). Second, the only useful information from $E_i^+$ is to tell which points in $S[1,i]$ are adjacent to the same component of $(S[i+1,n]_c,(E_i^+)_c)$, for each primary color $c\in \{1,\ldots ,k\}$. This information can be summarized by $k$ equivalence relations on the sets $(E_i^0)_1, \ldots , (E_i^0)_k$. We continue with the details.

We can encode $E_i^0$ by the set of its colors $\Gamma_i=\{\alpha(e): e\in E_i^0\}$.
For $i\in \{1,\ldots , n\}$, a set of edges $X_i^0$ between $S[1,i]$ and $S[i+1,n]$ is \emph{valid} if there exists a CSG $G=(S,E)$ such that $X_i^0=E_i^0$.
\begin{lemma}\label{lem:valid}
	For $i\in \{1,\ldots , n\}$, an edge set $X_i^0$ between $S[1,i]$ and $S[i+1,n]$ is \emph{valid} iff
    for every primary color $c\in \{1,\ldots, k\}$, there is an edge $e\in X_i^0$ such that $c\in \alpha(e)$
    whenever both $S[1,i]_c$ and $S[i+1,n]_c$ are nonempty.
\end{lemma}

We encode the relevant information from $E_i^+$ using $k$ equivalence relations as follows. For every $c\in \{1, \ldots, k\}$, the components of  $(S[i+1,n]_c,(E_i^+)_c)$ define an equivalence relation on $(E_i^0)_c$, which we denote by $\pi_i^c$: two edges in $(E_i^0)_c$ are related iff they are incident to the same component of $(S[i+1,n]_c,(E_i^+)_c)$. Let $\Pi_i=(\pi_i^1,\ldots , \pi_i^k)$. The equivalence relation $\pi_i^c$, in turn, determines a graph $(S[1,i]_c, E(\pi_i^c))$: two distinct vertices in $S[1,i]_c$ are adjacent iff they are incident to equivalent edges in $(E_i^0)_c$ (that is, two distinct vertices in $S[1,i]_c$ are adjacent iff they both are adjacent to the same component of $(S[i+1,n]_c,(E_i^+)_c)$). See Fig.~\ref{fig:collinear} for examples of $E^0_i$ and $\Pi_i$.
The condition that $(S,X_i^-\cup E_i^0\cup E_i^+)$ is a CSG can now be formulated in terms of $E_i^0$ and $\Pi_i$
(without using $E_i^+$ directly).

\begin{lemma}\label{lem:halves}
	Let $G=(S,E)$ be a CSG, $i\in \{1,\ldots, n\}$, and $X_i^-$ an edge set on $S[1,i]$.
	The graph $(S,X_i^-\cup E_i^0\cup E_i^+)$ is a CSG iff
	the graph $(S[1,i]_c,(X_i^-)_c\cup E(\pi_i^c))$ is connected for every $c\in \{1\ldots ,k\}$.
\end{lemma}

We can now define subproblems for \textsc{Min-$k$CSG}. For an index $i\in \{1,\ldots , n\}$, a valid set $E_i^0$, and equivalence relations $\Pi_i=(\pi_i^1,\ldots ,\pi_i^k)$, let $\mathcal{X}(E_i^0,\Pi_i)$ be the family of edge sets $X_i^-$ on $S[1,i]$ such that for every $c\in \{1\ldots ,k\}$, the graph $(S[1,i]_c,(X_i^-)_c\cup E(\pi_i^c))$ is connected.
The subproblem \textbf{A}$[i, E_i^0, \Pi_i]$ is to find the minimum cost of an edge set $X_i^-\in \mathcal{X}(E_i^0,\Pi_i)$.

Note that for $i=n$, \textbf{A}$[n, \emptyset, (\emptyset, \ldots, \emptyset)]$ is the minimum cost of a CSG for an instance $(S,\alpha)$ of $\textsc{Min-$k$CSG}$. Next, we establish a recurrence relation for \textbf{A}$[i, E^0_i, \Pi_i]$, which will allow computing \textbf{A}$[n, \emptyset, (\emptyset, \ldots, \emptyset)]$ by dynamic programming. For $i=1$, we have \textbf{A}$[1, E_1^0, \Pi_1]=0$ for any valid $E_1^0$ and $\Pi_1$. For all $i$, $1<i\leq n$, we wish to express \textbf{A}$[i, E_i^0, \Pi_i]$ in terms of \textbf{A}$[i-1, E_{i-1}^0, \Pi_{i-1}]$'s for suitable $E_{i-1}^0$ and $\Pi_{i-1}$.

We say that two valid edge sets $E_{i-1}^0$ and $E_i^0$ are \emph{compatible} if there exists an $X_i^-\in \mathcal{X}(E_i^0,\Pi_i)$ for some $\Pi_i$ such that $E_{i-1}^0 = (X_i^-\cup E_i^0)_{i-1}^0$. We can characterize compatible edge sets as follows.
\begin{lemma}\label{lem:compatible}
	Two valid edge sets $E_{i-1}^0$ and $E_i^0$ are compatible iff
	every edge $e$ in the symmetric difference of $E_{i-1}^0$ and $E_i^0$ is incident to $p_i$.
\end{lemma}

For two valid compatible edge sets, $E_{i-1}$ and $E_i$, and a sequence of equivalence relations $\Pi_i$, we define equivalence relations $\widehat{\Pi}_{i-1}=(\hat{\pi}_{i-1}^1,\ldots, \hat{\pi}_{i-1}^k)$ as follows. For every primary color $c\in \{1,\ldots, k\}$, let the equivalence relation $\hat{\pi}_{i-1}^c$ on $(E_{i-1}^0)_c$ be the transitive closure of the union of four equivalence relations: two edges in $(E_{i-1}^0)_c$ are related if (1) they both incident to $p_i$; (2) they both are in $(E_i^0)_c$ and $\pi_i^c$-equivalent; (3) they are both in $(E_i^0)_c$ and each are equivalent to some edge in $(E_i^0)_c$ that are $\pi_i^c$-equivalent; (4) one is incident to $p_i$ and the other is in $(E_i^0)_c$ and $\pi_i^c$-equivalent to some edge in $(E_i^0)_c$ incident to $\pi_i^c$.

\begin{lemma}\label{lem:part}
	Let $E_{i-1}^0$ and $E_i^0$ be two valid compatible edge sets, and $\Pi_i=(\pi_i^1,\ldots ,\pi_i^c)$.
	Let $E_{i-1}^-$ be a set of edges on $S[1,i-1]$, and put $E=E_{i-1}^-\cup E_{i-1}^0\cup E_i^0$.
	Then, $\widehat{\Pi}_{i-1}$ has the following property:
	$E_i^-\in \mathcal{X}(E_i^0,\Pi_i)$ if and only if
	\begin{enumerate}
		\item[(\textbf{d1})] $E_{i-1}^-\in \mathcal{X}(E_{i-1}^0,\widehat{\Pi}_{i-1})$; and
		\item[(\textbf{d2})] if $c\in \alpha(p_i)$ and $S[1,i]_c\neq \{p_i\}$, then $p_i$ is incident to
		an edge in $(E_{i-1}^0)_c$ or an edge in $(E_i^0)_c$ that is $\pi_i^c$-equivalent to some edge incident to $S[1,i-1]_c$.
	\end{enumerate}
\end{lemma}

\begin{figure}[htp]
	\centering
	\def\svgwidth{0.9\columnwidth}
	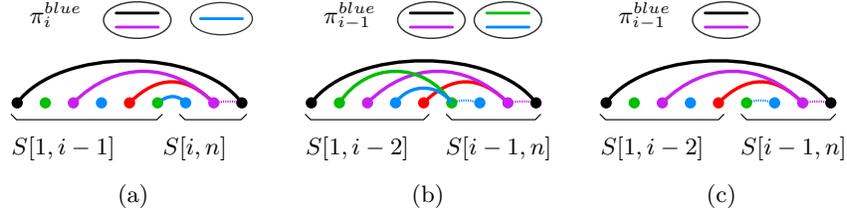
	\caption{(a) $E^0_i$ and $\pi_i^{blue}$. (b) $E^0_{i+1}$ and $\pi_{i+1}^{blue}$, where $E^0_{i+1}$ and $E^0_i$ are compatible.
 (c) $E^0_{i+1}$ and $\pi_{i+1}^{blue}$ violate condition (\textbf{d2}).}
	\label{fig:collinear}
\end{figure}

\begin{lemma} \label {lem:colcur}
	For all $i\in \{2,\ldots , n\}$, we have the following recurrence:
	\begin{equation}
		\label{eq:DynamicProg}
		\textbf{A}[i, E_i^0, \Pi_i]=\sum_{\{p_h,p_i\}\in E^0_i}w(\{p_h,p_i\})
		+\min_{E_{i-1}^0\text{\rm compatible}} \textbf{A}[i-1, E_{i-1}^0, \widehat{\Pi}_{i-1}] .
	\end{equation}
\end{lemma}

\begin{theorem}\label{thm:collinear-alg}
	For every constant $k\geq 1$, \textsc{Min-$k$CSG} can be solved in $O(n)$ time when the input points are collinear.
\end{theorem}
\begin{proof}
	We determine the number of subproblems. By Corollary~\ref{cor:unique-edge}, every valid
	$E_i^0$ contains at most $|2^{\{1,\ldots,k\}}\setminus\{\emptyset\}|=2^k-1$ edges.
	We have $|(E_i^0)_c|\leq 2^{k-1}$, since $2^{k-1}$ different colors contain any primary color $c\in\{1,\ldots,k\}$.
	The number of equivalence relations of a set of size $t$ is known as the $t$-th \emph{Bell number}, denoted $B(t)$.
	It is known~\cite{BT10} that $B(t)\leq (0.792t/\ln(t+1))^t<2^{O(t \log t)}$.
	Consequently, the number of possible $\Pi_i$ is at most $(B(2^{k-1}))^k$.
	The total number of subproblems is $O(n 2^k (B(2^{k-1}))^k)$, which is $O(n)$ for any constant $k$.
	We solve the subproblems  \textbf{A}$[i, E_i^0, \Pi_i]$, $1<i\leq n$, by dynamic programming,
	using the recursive formula \eqref{eq:DynamicProg}. The time required to evaluate~\eqref{eq:DynamicProg}
	is $O(2^k)$ for the sum of edge weights and $O(2^k (B(2^{k-1}))^k)$ to compare all compatible subproblems
	\textbf{A}$[i-1, E_{i-1}^0, \widehat{\Pi}_{i-1}]$, that is, $O(1)$ time when $k$ is a constant.
	Therefore, the dynamic programming can be implemented in $O(n)$ time.
\end{proof}

\section{Conclusions}

We have shown that \textsc{Min-3CSG} is NP-complete in general and given a $O(n)$ time algorithm for \textsc{Min-$k$CSG} in the special case that all points are collinear and $k$ is a constant.
We also improved the approximation factor of a polynomial time algorithm from $(2+\frac{1}{2}\varrho)$~\cite{HKK16} to $(2-\frac{2}{2+2\varrho})$ when $k=3$. It remains open whether there exists a PTAS for \textsc{Min-$k$CSG}, $k\ge 3$.
Several other special cases are open for \textsc{Min-3CSG}, such as when the points in $S$ are on a circle or in convex position.
We can generalize \textsc{Min-$k$CSG} so that the edge weights need not be Euclidean distances. Given an arbitrary graph $(V,E)$ and a coloring $\alpha:V\rightarrow \mathcal{P}(\{1,\ldots,k\})$, what is the minimum set $E'\subseteq E$ such that $(V,E')$ is a colored spanning graph?
Since the 2-approximation algorithm presented here did not rely on the geometry of the problem, it extends to the generalization; however, this problem may be harder to approximate than its Euclidean counterpart.


\smallskip\noindent{\bf Acknowledgements.}
Research on this paper was supported in part by the NSF awards CCF-1422311 and CCF-1423615. Akitaya was supported by the Science Without Borders program. L{\"o}ffler was partially supported by the Netherlands Organisation for Scientific Research (NWO) projects 639.021.123 and 614.001.504.

\newpage

\appendix

\section{Omitted proofs} \label {app:proofs}

\subsection {Proofs from Section~\ref{sec:hard}}


	\begin{figure}[htp]
		\centering
		\def\svgwidth{0.8\columnwidth}
		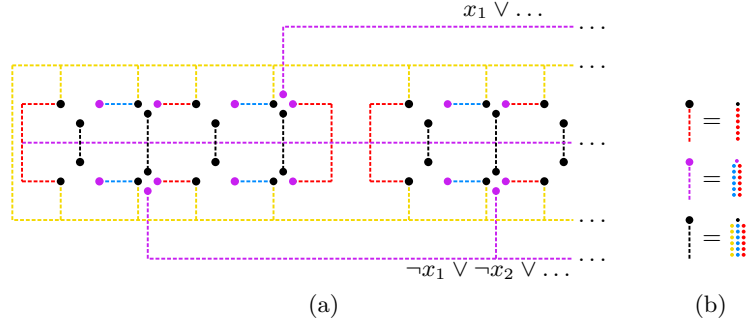
		\caption{(a) Reduction from \textsc{Planar-Monotone-3SAT} to \textsc{Min-3CSG}. Dashed lines correspond to series of $\eps$ close points as shown in (b).}
		\label{fig:reduction}
		\centering
	\end{figure}	
	
	\begin{repeatlemma}{lem:sat2csg}
		Let $A$ be a positive instance of \textsc{Planar-Monotone-3SAT}. Then $f(A)$ is a positive instance of \textsc{Min-3CSG}.
	\end{repeatlemma}
	
	\begin{proof}
		We construct a standard solution for $f(A)$ based on the solution for $A$.
		First connect all points along each dashed line using $\sum_{e\in E'}w(e))+39r\eps$ length.
		If a variable is set to \texttt{true} in $A$, connect the active points of all 2-switches of the corresponding variable gadget as shown in Fig.~\ref{fig:switch}(a). Otherwise connect them as the reflection of Fig.~\ref{fig:switch}(a) about the $x$-axis.
		The length of these edges sum to $2r(1+\sqrt{2})$.
		For each positive (resp., negative) clause, choose an arbitrary neighbor variable assigned \texttt{true}  (resp., \texttt{false}) and connect the active points in the $2\delta$-switch as is shown in Fig.~\ref{fig:switch}(d) (resp., reflection of Fig.~\ref{fig:switch}(d)).
		The length of such edges sum to $4m\delta\sqrt{2}$.
		Connect all remaining $2\delta$-switches as shown in Figs.~\ref{fig:switch}(c) (or its reflection) depending on the assignment of the neighbor switch.
		The length of such edges sum to $2(r-m)\delta(1+\sqrt{2})$.
		The resulting graph is a colored spanning graph and its total weight is $W$.
	\end{proof}

\begin{figure}[h]
	\centering	
	\includegraphics[width=0.75\linewidth]{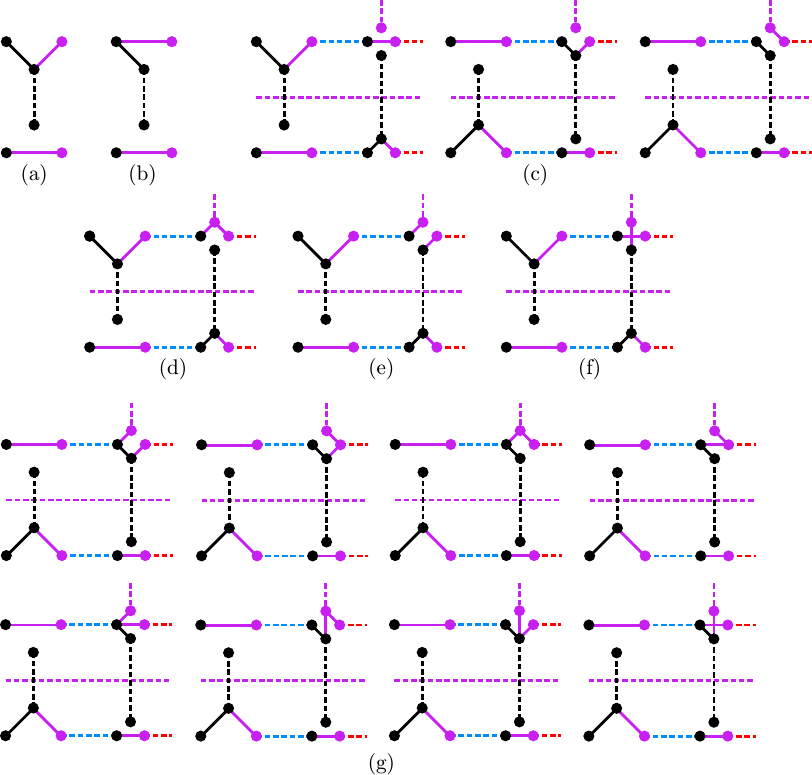}
	\caption{All possible ways to connect the vertices of a switch in a standard solution satisfying all constraints with the minimum number of edges: (a) and (b) a pair of 2-switches with three edges; (c) a pair of $2\delta$-switches with three edges; (d), (e), (f), and (g) a pair of $2\delta$-switches with four edges such that there is a blue path from the clause gadget to the spine that uses an edge from the switches.}
	\label{fig:switch}
\end{figure}

	\begin{repeatlemma}{lem:existance}
		If $f(A)$ is a positive instance of \textsc{Min-3CSG}, there exists a standard solution for this instance.
	\end{repeatlemma}
	
	\begin{proof}
		By Lemma~\ref{lem:monocrhomatic}, if $f(A)$ is a positive instance of \textsc{Min-3CSG}, there exists a solution containing $E'$.
		We consider only such solutions.
		Suppose that there exists an edge
		$e=\{q_1,q_2\}$ such that $w(\{q_1,q_2\})>\eps$,  and  $q_1$  or $q_2$ is not an active point.
		Since only active points are multi-chromatic, $e$ is necessarily monochromatic, therefore, its removal can only affect the connectivity of the color $c$, where $\{c\}=\alpha(\{q_1,q_2\})$, disconnecting the induced graph into two connected components.
		Since this solution includes $E'$, by property (I) there exists an edge of length $\eps$ that reconnects the components or there exists an edge $e'=\{p_1,p_2\}$ such that $w(e')<w(e)$.
		In both cases, $e$ can be replaced by such an edge, obtaining a lighter solution.
		It remains to consider the case that $q_1$ and $q_2$ are active points, but in different switches.
		In that case, $w(\{q_1,q_2\})>6$ by construction. The deletion of $\{q_1,q_2\}$ may disconnect
        up to three graphs induced by primary colors.
		For each primary color $c$, the two components are either $\eps$ apart or there exists a switch that contains active points belonging to both components, by property (I).
		Hence, replacing $\{q_1,q_2\}$ by at most three edges, each of length at most 2, produces a solution of equal or smaller cost.
	\end{proof}

\begin{figure}
	\centering
	\includegraphics[width=0.75\linewidth]{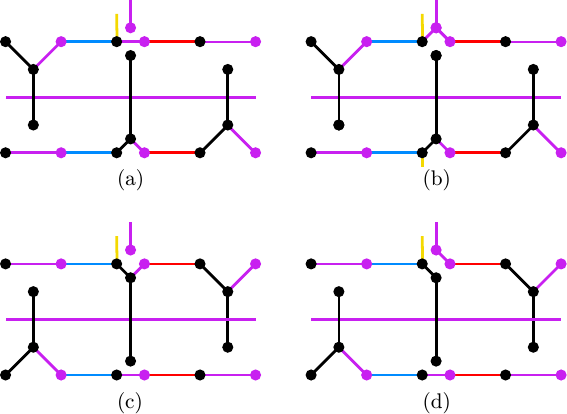}
	\caption{Counter-examples that show that two consecutive 2-pairs of the same variable gadget cannot have different orientations.}
	\label{fig:conter-expamples}
\end{figure}

	\begin{repeatlemma}{lem:min}
		In a standard solution, the minimum length required to satisfy the local constraints of a $2$-pair (resp., $2\delta$-pair) is $2(1+\sqrt{2})$ (resp., $2\delta(1+\sqrt{2})$).
	\end{repeatlemma}
	
	\begin{proof}
		To satisfy the rib constraint, assume without loss of generality that the upper switch contain a black edge.
		Since the switch constraints of the pair still require at least two more edges, the solution must have at least three edges.
		We can enumerate all possible local solutions that satisfy the switch constraints using a total of three edges (Fig.~\ref{fig:switch}(a) and (b)).
		Notice that every edge in the switch is $2$ or $\sqrt{2}$ long ($2\delta$ or $\sqrt{2}\delta$ for $2\delta$-pairs).
		Therefore, all possible solutions that use four edges have a greater cost than the ones using three edges.
		Then, the local solution with minimum cost must be as shown in Fig.~\ref{fig:switch}(a).
		Notice that this lower bound continues to hold in the presence of an active point of a clause.
        However, if an active point of a clause is part of a switch, the local solution with minimum
        cost is no longer unique; Fig.~\ref{fig:switch}(c) shows all solutions in that case.
	\end{proof}

	\begin{repeatcorollary}{cor:2-pair}
		In a standard solution, every 2-pair is connected minimally.
	\end{repeatcorollary}
	
	\begin{proof}
		For contradiction, suppose that there exists a 2-pair whose local cost is at least $4+\sqrt{2}$ (the second minimum is depicted in Fig.~\ref{fig:switch}(b)).
		Then, by Lemma~\ref{lem:min} the cost of such a solution must be at least $\sum_{e\in E'}w(e)+(r-1)(2(1+\sqrt{2}))+4\sqrt{2}=W+2\sqrt{2}-[2+39r\eps+2\delta(r(1+\sqrt{2})+m(\sqrt{2}-1))]$ which is greater than $W$ because $2m\le r\le 3m$ and by the choice of $\eps$ and $\delta$.
	\end{proof}

	\begin{repeatlemma}{lem:clause}
		In a standard solution, for each clause gadget, there exists a $2\delta$-pair with local cost at least $4\delta\sqrt{2}$.
	\end{repeatlemma}
	
	\begin{proof}
		Each clause gadget requires at least a blue path between one of its active points and the spine, or else the subgraph induced by $S_\texttt{b}$ is disconnected.
		Consider the $2\delta$-pair that contains an edge in such a path.
		We can enumerate all local solutions that satisfy the local constraints and contain a blue path between the clause active point and the spine and that uses 4 edges in the $2\delta$-pair (Fig.~\ref{fig:switch}(d)--(g)).
		By Corollary~\ref{cor:2-pair}, we consider only solutions in which $2$-pairs are minimally connected.
		Notice that removing any of the edges violates some local constraint, hence, at least 4 edges are required.
		Also notice that every local solution that uses 5 or more edges have cost $5\delta\sqrt{2}$ or greater.
		The local solutions with minimum cost are shown in Figs.~\ref{fig:switch}(d) and (e) and their cost is $4\delta\sqrt{2}$.
		
	\end{proof}
	
	\begin{repeatcorollary}{cor:clause}
		In a standard solution, for each clause gadget, there exists a $2\delta$-pair connected as Fig.~\ref{fig:switch}(d). All other $2\delta$-pairs are connected minimally as shown in Fig.~\ref{fig:switch}(c).
	\end{repeatcorollary}
	
	\begin{proof}
		By Lemma~\ref{lem:clause}, the lower bound for the cost of $m$ of the $2\delta$-pairs is $4m\delta\sqrt{2}$.
		For contradiction, suppose that one of the  $2\delta$-pairs uses more length than its lower bound.
		The second minimal configuration on the switches occurs when $m$ distinct $2\delta$-pairs are connected using $4m\delta\sqrt{2}$ length, $r-m-1$ of the remaining $2\delta$-pairs are connected using $(r-m-1)\delta(2+2\sqrt{2})$, and the remaining $2\delta$ pair uses $4\delta\sqrt{2}$. 		
		This construction costs at least $\sum_{e\in E'}w(e)+r(2+2\sqrt(2))+(r-m)\delta(2+4\sqrt{2})+2\sqrt{2}\delta-2\delta=W+2\sqrt{2}\delta-(39r\eps+2\delta)$ which is greater than $W$ by the choice of $\eps$ and $\delta$.
		Hence, $r-m$ of the $2\delta$-pairs have to be connected as shown in Fig.~\ref{fig:switch}(c) and $m$ of them have to be connected as in Fig.~\ref{fig:switch}(d) and (e).
		However, if one of these $2\delta$-pairs is connected as Fig.~\ref{fig:switch}(e), then there exists no local red path between the corresponding clause gadget and the spine.
		If we assume that such a path passes through a different switch, it must be part of a minimally connected $2\delta$-pair, since we can only afford any $2\delta$-pair with $4\delta\sqrt{2}$ cost per clause.
		The only option would be to use the edges shown in Fig.~\ref{fig:conter-expamples}(d), which is impossible because then the graph induced by $S_\texttt{r}$ would be disconnected (notice that in Fig.~\ref{fig:conter-expamples}(d) there exists a red component in the bottom that is not connected to the spine). Hence, all $m$ these $2\delta$-pairs must be as shown in Fig.~\ref{fig:switch}(d).
	\end{proof}

	\begin{repeatlemma}{lem:csg2sat}
		Let $f(A)$ be a positive instance of \textsc{Min-3CSG}. Then $A$ is a positive instance of \textsc{Planar-Monotone-3SAT}.
	\end{repeatlemma}
	
	\begin{proof}
		We say that a 2-pair is set to \texttt{true} if its upper switch contains a black edge and is set to \texttt{false} otherwise.
		First we show that every $2$-pair of the same variable gadget is set to the same value.
		For contradiction assume that two consecutive 2-pairs are set to different truth values.
		Fig.~\ref{fig:conter-expamples} shows all possible configurations that satisfy Corollaries~\ref{cor:2-pair} and \ref{cor:clause}.
		As already stated in the proof of Corollary~\ref{cor:clause}, the local configuration in Fig.~\ref{fig:conter-expamples}(d) is excluded.
		The configurations in Figs.~\ref{fig:conter-expamples}(a) and (c) lead to a red induced component on the top and bottom part of the construction, respectively, hence $S_\texttt{r}$ would induce a disconnected graph.
		The configuration in Fig.~\ref{fig:conter-expamples}(b) would also imply that $S_\texttt{r}$ induces a disconnected graph, since there would not be any red induced path between the corresponding clause and the spine.
		
		We conclude that any standard solution is similar to the one described in the proof of Lemma~\ref{lem:sat2csg}. Then, we can easily assign a truth value from $ \{\texttt{true},  \texttt{false}\} $ to every variable, obtaining a solution for the \textsc{Planar-Monotone-3SAT}.
	\end{proof}

\subsection {Proofs from Section~\ref{sec:collinear}}

\begin{repeatlemma}{lem:valid}
	For $i\in \{1,\ldots , n\}$, an edge set $X_i^0$ between $S[1,i]$ and $S[i+1,n]$ is \emph{valid} iff
    for every primary color $c\in \{1,\ldots, k\}$, there is an edge $e\in X_i^0$ such that $c\in \alpha(e)$
    whenever both $S[1,i]_c$ and $S[i+1,n]_c$ are nonempty.
\end{repeatlemma}
\begin{proof}
Let $G=(S,E)$ be a CSG with property ($\star$). Then for every $c\in\{1,\ldots , k\}$, the graph $(S_c,E_c)$ is connected. If both $S[1,i]_c$ and $S[i+1,n]_c$ are nonempty, there is an edge $e\in (E_0^i)_c$, hence $c\in \alpha(e)$. 

Conversely, assume every edge in $X_i^0$ has property ($\star$). Let $E_i^-$ (resp., $E_i^+$) be the set of all edges on $S[1,i]$ (resp., $S[i+1,n]$) with property ($\star$). By Lemma~\ref{lem:unique-color}, $(S[1,i]_c,(E_i^-)_c)$ and $(S[i+1,n])c,(E_i^+)_c)$ have the same components as the compete graph in each primary color $c\in \{1,\ldots , k\}$. Then
$(E_i^0)_c$ contains at least one edge between $S[1,i]_c$ and $S[i+1,n]_c$ if both are nonempty.
Consequently, $(S,E_1^-\cup E_i^0\cup E_i^-)$ is a CSG for $S$, as required.
\end{proof}

\begin{repeatlemma}{lem:halves}
	Let $G=(S,E)$ be a CSG, $i\in \{1,\ldots, n\}$, and $X_i^-$ an edge set on $S[1,i]$.
	The graph $(S,X_i^-\cup E_i^0\cup E_i^+)$ is a CSG iff
	the graph $(S[1,i]_c,(X_i^-)_c\cup E(\pi_i^c))$ is connected for every $c\in \{1\ldots ,k\}$.
\end{repeatlemma}
\begin{proof}
	Let $c\in \{1\ldots ,k\}$. We claim that $(S_c,(X_i^-\cup E_i^0\cup E_i^+)_c)$ is connected iff
	$(S[1,i]_c,(E_i^-)_c\cup E(\pi_i^c))$ is connected.
	If $S[1,i]_c=\emptyset$, then $(S_c,(X_i^-\cup E_i^0\cup E_i^+)_c)=(S[i+1,n]_c,(X_i^-)_c)$.
	If $S[i+1,n]_c=\emptyset$, then $(S_c,(X_i^-\cup E_i^0\cup E_i^+)_c)=(S[1,i]_c,(E_i^+)_c)$.
	In both cases, the claim follows.
	
	Assume that neither $S[1,i]_c$ nor $S[i+1,n]_c$ is empty. Since $G$ is a CSG, then $(S_c,E_c)$ is connected, and so each component of $(S[i+1,n]_c,(E_i^+)_c)$ is incident to some edge in $(E_i^0)_c$ that has an endpoint $S[1,i]_c$. Consequently, each component of $(S_c,(X_i^-\cup E_i^0\cup E_i^+)_c)$ has a vertex in $S[1,i]_c$. Therefore $(S_c,(X_i^-\cup E_i^0\cup E_i^+)_c)$ is connected iff it contains a path between any two vertices of $S[1,i]_c$. Two vertices in $S[1,i]_c$ are connected by a path in $(E_i^0\cup E_i^+)_c$ iff they are adjacent along an edge in $E(\pi_i^c)$. Consequently, there is a path between any two vertices in $S[1,i]_c$ iff they are in the same component of
	$(S[1,i]_c,(E_i^-)_c\cup E(\pi_i^c))$, as required.
\end{proof}

\begin{repeatlemma}{lem:compatible}
	Two valid edge sets $E_{i-1}^0$ and $E_i^0$ are compatible iff
	every edge $e$ in the symmetric difference of $E_{i-1}^0$ and $E_i^0$ is a incident to $p_i$.
\end{repeatlemma}
\begin{proof}
	Assume that $E_{i-1}^0$ is compatible with $E_i^0$, witnessed by some $X_i^-\in \mathcal{X}(E_i^0,\Pi_i)$. Property ($\star$) implies that the left endpoint of an edge $e\in E_i^0$ is $p_i$ iff $\alpha(e)\subset \alpha(p_i)$.
	Similarly, the right endpoint of an edge in $E_{i-1}^0$ is $p_i$ iff $\alpha(e) \alpha(p_i)$.
	
	Conversely, assume every edge $e$ in the symmetric difference of $E_{i-1}^0$ and $E_i^0$ is a incident to $p_i$. Let $E_{i-1}^-$ be the set of all edges on $S[i-1]$ with property ($\star$); and let $X=E_{i-1}^-\cup E_{i-1}^0\cup E_i^0$. Define $\Pi_i$ such that all members of $(E_i^0)_c$ are $\pi_i^c$ equivalent for every $c\in \{1,\ldots, k\}$.  By Lemma~\ref{lem:unique-color}, the graph $(S[1,i-1]_c,(E_{i-1}^-)_c)$ is connected for all primary colors $c\in \{1,\ldots ,k\}$.
	
	For every color $c\in \{1,\ldots ,k\}\setminus \alpha(p_i)$, we have $S[1,i]_c=S[1,i-1]_c$ and so $(S[1,i]_c,(X_i^-)_c\cup E(\pi_i^c))$ is connected.
	Let $c\in \alpha(p_i)$. If $S[1,i]_c=\{p_i\}$, then $(S[1,i]_c,\emptyset)$ is clearly connected. Otherwise, $p_i$ is incident to an edge in $(E_{i-1}^0)_c$ or $(E_i^0)_c$, by Lemma~\ref{lem:valid}. If $p_i$ is incident to an edge in $(E_{i-1}^0)_c$, then it is adjacent a vertex in $S[1,i-1]_c$. Otherwise $p_i$ is incident to some $e_1\in (E_i^0)_c$, consequently $S[i+1,n]_c\neq\emptyset$. By Lemma~\ref{lem:valid}, there is some $e_2\in (E_{i-1}^0)_c\cap (E_i^0)_c$ between $S[1,i-1]_c$ and $S[i+1,n]_c$. By construction, $e_1$ and $e_2$ are $\pi_i^c$-equivalent, and so $(S[1,i]_c,(X_i^-)_c\cup E(\pi_i^c))$ is connected, as required.
\end{proof}

\begin{repeatlemma}{lem:part}
	Let $E_{i-1}^0$ and $E_i^0$ be two valid compatible edge sets, and $\Pi_i=(\pi_i^1,\ldots ,\pi_i^c)$.
	Let $E_{i-1}^-$ be a set of edges on $S[1,i-1]$, and put $E=E_{i-1}^-\cup E_{i-1}^0\cup E_i^0$.
	Then, $\widehat{\Pi}_{i-1}$ has the following property:
	$E_i^-\in \mathcal{X}(E_i^0,\Pi_i)$ if and only if
	\begin{enumerate}
		\item[(\textbf{d1})] $E_{i-1}^-\in \mathcal{X}(E_{i-1}^0,\widehat{\Pi}_{i-1})$; and
		\item[(\textbf{d2})] if $c\in \alpha(p_i)$ and $S[1,i]_c\neq \{p_i\}$, then $p_i$ is incident to
		an edge in $(E_{i-1}^0)_c$ or an edge in $(E_i^0)_c$ that is $\pi_i^c$-equivalent to some edge incident to $S[1,i-1]_c$.
	\end{enumerate}
\end{repeatlemma}
\begin{proof}
	Assume $E_i^-\in \mathcal{X}(E_i^0,\Pi_i)$.
	Then the graph $(S[1,i]_c,(E_i^-)_c\cup E(\pi_i^c))$ is connected for every $c\in \{1,\ldots, k\}$, however, its induced subgraph on $S[1,i-1]_c$ may have two or more components. Between any two such components, there is a path in $(E_i^-\setminus E_{i-1}^-)_c\cup E(\pi_i^c)$, where
	$(E_i^-\setminus E_{i-1}^-)$ is the set of edges whose right endpoint is $p_i$.
	The first and last edge of any such path is $\hat{\pi}_{i-1}^c$-equivalent by the definition of $\pi_{i-1}^c$.
	Consequently, $E_{i-1}^-\in \mathcal{X}(E_{i-1}^0,\widehat{\Pi}_{i-1})$, and (\textbf{d1}) follows.
	For every $c\in \alpha(p_i)$, if $S[1,i-1]_c\neq \{p_{i-1}\}$, then $p_i$ is incident to some edge in $(E_{i-1}^-)_c\cup E(\hat{\pi}_i^c)$ by (\textbf{d1}), which implies (\textbf{d2}).
	
	Conversely, assume (\textbf{d1}) and (\textbf{d2}). Since $E_{i-1}^0\in \mathcal{X}(E_i^0,\widehat{\Pi}_{i-1})$, then the graph $(S[1,i]_c,(E_{i-1}^-)_c\cup E(\hat{\pi}_{i-1}^c))$ is connected for every $c\in \{1,\ldots, k\}$. That is, there is a path in $(E_{i-1}^-)_c\cup E(\hat{\pi}_{i-1}^c)$ between any two vertices in $S[1,i-1]_c$. If we replace the edges in $E(\hat{\pi}_{i-1}^c)$ with a sequence of edges thate in $E(\pi_i^0)$ or incident to $p_i$ (using the definition of $\hat{\pi}_{i-1}^c$), there is a path in $(E_i^-)_c\cup E(\pi_i^c)$ between any two vertices in $S[1,i-1]_c$. Finally, for every $c\in \alpha(p_i)$, we need to show that there is a path in $(E_i^-)_c\cup E(\pi_i^c)$ between $p_i$ and any other vertex of $S[1,i]_c$. This is vacuously true if $p_i$ is the only vertex of $S[1,i]_c$. Suppose that $S[1,i]_c\neq \{p_i\}$. By (\textbf{d2}), $p_i$ is incident to an edge in $(E_{i-1}^0)_c$ or an edge in $E(\pi_i^c)$. Consequently, $E_i^-\in \mathcal{X}(E_i^0,\Pi_i)$, as required.
\end{proof}

\begin{repeatlemma}{lem:colcur}
	For all $i\in \{2,\ldots , n\}$, we have the following recurrence:
	\begin{equation}
		\label{eq:dynamicProg}
		\textbf{A}[i, E_i^0, \Pi_i]=\sum_{\{p_h,p_i\}\in E^0_i}w(\{p_h,p_i\})
		+\min_{E_{i-1}^0\text{compatible}} \textbf{A}[i-1, E_{i-1}^0, \widehat{\Pi}_{i-1}] .
	\end{equation}
\end{repeatlemma}

\begin {proof}
  The lemma follows directly from Lemmata~\ref{lem:compatible} and~\ref {lem:part}.
\end {proof}

\end{document}